\documentclass{amsart}
\usepackage{amssymb,amscd,amsthm}
\usepackage{latexsym}
\date{}

\newcommand{\Z}{{\mathbb Z}}
\newcommand{\R}{{\mathbb R}}
\newcommand{\C}{{\mathbb C}}
\newcommand{\N}{{\mathbb N}}
\newcommand{\T}{{\mathbb T}}

\newcommand{\E}{{\mathbb E}}

\newtheorem{theorem}{Theorem}
\newtheorem{remark}{Remark}[section]
\newtheorem{lemma}[remark]{Lemma}
\newtheorem{prop}[remark]{Proposition}
\newtheorem{coro}[remark]{Corollary}

\begin{document}
\title[A Continuum Version of the Kunz-Souillard Approach]{A Continuum Version of the Kunz-Souillard Approach to Localization in One Dimension}

\author{David Damanik}

\address{Department of Mathematics, Rice University, Houston, TX~77005, USA}

\email{damanik@rice.edu}

\thanks{D.\ D.\ was supported in part by NSF grants DMS--0653720 and DMS--0800100.}

\author{G\"unter Stolz}

\address{Department of Mathematics, University of Alabama at Birmingham, Birmingham, AL~35294, USA}

\email{stolz@math.uab.edu}

\thanks{G.\ S.\ was supported in part by NSF grant DMS-0653374.}

\begin{abstract}
We consider continuum one-dimensional Schr\"odinger operators with potentials that are given by a sum of a suitable background potential and an Anderson-type potential whose single-site distribution has a continuous and compactly supported density. We prove exponential decay of the expectation of the finite volume correlators, uniform in any compact energy region, and deduce from this dynamical and spectral localization. The proofs implement a continuum analog of the method Kunz and Souillard developed in 1980 to study discrete one-dimensional Schr\"odinger operators with  potentials of the form background plus random.
\end{abstract}

\maketitle

\section{Introduction}

In their work \cite{ks}, Kunz and Souillard introduced the first
method which allowed one to give a rigorous proof of localization for
the one-dimensional discrete Anderson model; see also \cite{cfks}
for a presentation of the method. Since then, other, more
versatile, methods have been developed which allowed one to establish
Anderson localization also for multi-dimensional and continuum
Anderson models. However, the Kunz-Souillard method has several
important features which still deserve interest. First, it
directly establishes a strong form of dynamical localization which
was proven with other methods only much later and with more effort.

Second, and most importantly, among the available methods which
establish one-dimensional localization, it is the only one which
shows localization at all energies, arbitrary disorder, and
without requiring ergodicity. It completely avoids the use of
Lyapunov exponents and provides an extremely direct path to
dynamical localization.

The virtues of the Kunz-Souillard methods were demonstrated in the
work \cite{desiso}, which proves localization for one-dimensional
discrete Anderson models with arbitrary bounded background
potential, and in \cite{simon}, which applies the method to
decaying random potentials.

It is our goal here to extend the Kunz-Souillard method to
one-dimensional {\it continuum} Anderson-type models and provide a
localization proof which allows for a rather general class of
bounded, not necessarily periodic, background potentials.

We mention one earlier work of Royer \cite{r} which carried over
many of the features of the Kunz-Souillard method to a continuum
model. There the Brownian motion generated random potential from
the work \cite{gmp} by Goldsheid, Molchanov and Pastur is considered. The
argument in \cite{r} does not carry over to Anderson-type models
and does not include deterministic background potentials.

Specifically, we consider the random Schr\"odinger operator
\begin{equation} \label{model}
H_{\omega} = -\frac{d^2}{dx^2} + W_0(x) + V_{\omega}(x)
\end{equation}
in $L^2(\R)$, where the random potential is given by
\begin{equation} \label{randompot}
V_{\omega}(x) = \sum_{n\in \Z} \omega_n f(x-n).
\end{equation}
For the single-site potential $f$, the random
coupling constants $\omega_n$, and the background potential $W_0$,
we fix the following assumptions:

The coupling constants $ \omega = (\omega_n)_{n\in\Z}$ are i.i.d.\
real random variables, whose distribution has a continuous and
compactly supported density $r$.

For the single-site potential $f$, we assume that
\begin{equation} \label{singlesite1}
c \chi_I(x) \le f(x) \le C \chi_{[-1,0]}
\end{equation}
for constants $0< c \le C < \infty$ and a non-trivial subinterval
$I$ of $[-1,0]$. We also assume that there exists an interval
$[a,b] \subset [-1,0]$ such that
\begin{equation} \label{singlesite2}
f(x)>0 \:\mbox{for a.e.}\: x\in [a,b], \quad f(x)=0 \:\mbox{for}\: x\in [-1,0]\setminus [a,b].
\end{equation}

The background potential $W_0$ is bounded and such that
\begin{equation} \label{relcomp}
\{ W_0(\cdot-n)|_{[-1,0]}: n\in \Z\} \: \mbox{is relatively compact in}\: L^{\infty}(-1,0).
\end{equation}

This includes 1-periodic potentials $W_0$, but, much more
generally, allows for many situations where $H_{\omega}$ is not
ergodic with respect to 1-shifts, examples being

(i) almost periodic potentials $W_0$, that is, $\{W_0(\cdot-\tau):\tau \in \R\}$ is relatively compact in $L^{\infty}(\R)$,

(ii) any $W_0$ that is uniformly continuous on $\R$, in which case \eqref{relcomp} follows from equicontinuity and the Arzela-Ascoli theorem.

Let $\chi_x = \chi_{[x-1,x]}$ and, for $x,y \in \Z$ and $E_{max} >0$, define
\begin{align}
\label{correlators} \rho(x,y;E_{max}) & := \E \Big( \sup\{
\|\chi_x g(H_{\omega}) \chi_y\|: \, g:\R\to \C \:\mbox{Borel
measurable}, \\
\nonumber & \qquad |g|\le 1, \: \mbox{supp}\,g \subset
[-E_{max},E_{max}]\} \Big).
\end{align}

Our main result is

\begin{theorem} \label{main}
For every $E_{max} >0$, there exist $C<\infty$ and $\eta>0$ such that for all $x,y \in \Z$,
\begin{equation} \label{dynloc}
\rho(x,y;E_{max}) \le Ce^{-\eta|x-y|}.
\end{equation}
\end{theorem}

Let us discuss some immediate consequences of this result. First, Theorem~\ref{main} implies a strong form of dynamical localization. Indeed, choosing $g_t(E) = e^{-itE} \chi_{[-E_{max},E_{max}]}(E)$, we find that
$$
\E \Big( \sup_{t \in \R} \|\chi_x e^{-i t H_{\omega}} \chi_{[-E_{max},E_{max}]}(H_{\omega}) \chi_y\| \Big) \le Ce^{-\eta|x-y|}.
$$
Dynamical localization implies spectral localization. Explicitly, we can deduce the following result.

\begin{coro} \label{specloc}
For almost every $\omega$, $H_\omega$ has pure point spectrum with exponentially decaying eigenfunctions.
\end{coro}

We refer the reader to \cite{aenss} for an explicit derivation of this consequence and other interesting forms and signatures of localization from a statement like Theorem~\ref{main}. Note that spectral localization throughout the spectrum is obtained by taking a countable intersection over an unbounded sequence of values of $E_{max}$.

Spectral localization and a weaker form of dynamical localization with periodic background potential can be established by different methods for more general single-site distributions. The most general result, which merely assumes non-trivial bounded support, can be found in \cite{dss}.

It has also been shown recently in \cite{hss} how the fractional moment method can be used to show the exponentially decaying dynamical localization bound \eqref{dynloc} for the model \eqref{model}, \eqref{randompot}. This work requires periodicity of the background potential and the existence of a bounded and compactly supported density for the single-site distribution.

Given these results for the case of periodic background and the limitation of the methods used in their proofs to this particular choice of background, we see that the method we present here should be regarded as the primary tool to perturb a given Schr\"odinger operator by the addition of a (preferably small) random potential to generate pure point spectrum. In fact, a question of this kind posed by T.\ Colding and P.\ Deift triggered our work. Thus, as a sample application, we state the following corollary, which answers the specific question we were faced with.

\begin{coro} \label{ColdingDeift}
Given $W_0 : \R \to \R$ with $C^1$-norm $\|W_0\|_{\infty} + \|W_0'\|_{\infty} < \infty$, there exists a $V : \R \to \R$
with arbitrarily small $C^1$-norm such that $-d^2/dx^2 + W_0 +V$ has pure point spectrum.
\end{coro}

This is an immediate consequence of Theorem~\ref{main} and Corollary~\ref{specloc}, since $W_0$ is uniformly continuous and we may choose $V$ of the form (\ref{randompot}) with single-site potential $f$ of small $C^1$-norm and $\omega_n$ supported in $[0,1]$.

Apart from the ability to handle quite general background potentials, the flexibility of the Kunz-Souillard method manifests itself in the discrete case in another way: it allows one to prove powerful results for random decaying potentials; see Simon \cite{simon} and the follow-up paper \cite{dss2}. We have not been able to establish a continuum analogue of Simon's work using our method. The obstacle to doing this is in making Proposition~\ref{t1normthm} in Section~\ref{s.T122strongbound} quantitative. We regard it as an interesting open problem to establish such a quantitative estimate. \footnote{Incidentally, it was realized sixteen years after Simon's 1982 work that there is a different approach to random decaying potentials which is applicable in more general situations and which does extend easily to the continuum \cite{kls}. Another approach to continuum random decaying potentials was given in \cite{ku}.}

The remainder of the paper is organized as follows. Section~\ref{s.reduction} describes the overall strategy of the proof of Theorem~\ref{main}, adapting the strategy of \cite{ks} and culminating in Section~\ref{ss.endofproof}. In the remaining sections we prove various norm bounds on integral operators which are used in this argument, with the crucial contraction property being established in Section~\ref{s.T122strongbound}. A technical result on the large coupling limit of Pr\"ufer amplitudes, used in Section~\ref{s.T122strongbound}, is proven in Appendix~\ref{s.appA}, while Appendix~\ref{s.appB} summarizes a number of standard ODE facts.

\vspace{.3cm}

\noindent {\bf Acknowledgements:} We would like to thank B.~Simon for useful discussions at an early stage of this work, as well as T.~Colding and P.~Deift for posing a question which motivated us to finish it.

\section{Reduction to Integral Operator Bounds} \label{s.reduction}

\subsection{Finite Volume Correlators} \label{ss:fvcorrelators}

Consider finite-volume restrictions of $H_\omega$, that is, for $L
\in \Z_+$, we denote by $H_\omega^L$ the restriction of $H_\omega$
to the interval $[-L,L]$ with Dirichlet boundary conditions. For fixed $L$ we will use the abbreviation $\omega = (\omega_{-L+1}, \ldots, \omega_L)$, as these are the only coupling constants which $H_{\omega}^L$ depends on.
Moreover, for $x,y \in \Z$ and $E_{max} >0$, we introduce the {\it finite volume correlators}
$$
\rho_L(x,y;E_{max}) := \E \Big( \sup\{ \|\chi_x g(H_{\omega}^L)
\chi_y\|: \, |g|\le 1, \: \mbox{supp}\,g \subset
[-E_{max},E_{max}]\} \Big).
$$

\begin{lemma}
We have
$$
\rho(x,y;E_{max}) \le \liminf_{L \to \infty} \rho_L(x,y;E_{max}).
$$
\end{lemma}

\begin{proof}
See \cite[Eq.~(2.28)]{aenss} and its discussion there.
\end{proof}

We will often suppress the dependence of these quantities on
$E_{max}$ in what follows. Moreover, we can without loss of
generality restrict our attention to the case $x = 1$ and $y = n
\in \Z_+$. Thus, we aim to estimate $\rho(1,n)$ by means of
finding estimates for $\rho_L(1,n)$ that are uniform in $L$.
Explicitly, our goal is to show the following:

\begin{prop}\label{l.finitevolkeylemma}
There exist $C<\infty$ and $\eta>0$ such that, for all $n,L \in
\Z_+$ with $n \le L$, we have
\begin{equation} \label{dynloc2}
\rho_L(1,n) \le Ce^{-\eta n}.
\end{equation}
\end{prop}

In order to estimate $\rho_L(1,n)$, we consider the eigenfunction
expansion of $H_\omega^L$. Thus, for $L
\in \Z_+$ and $\omega = (\omega_{-L+1}, \ldots, \omega_L)$, we denote the (simple) eigenvalues of $H_\omega^L$ by
$\{ E_k : k \ge 1 \}$ and the associated normalized eigenvectors
by $\{ v_k : k \ge 1\}$. Here we leave the dependence of these
quantities on $\omega$ and $L$ implicit.

The proof of Proposition~\ref{l.finitevolkeylemma} starts with the
following observation:
\begin{align}
\label{e.rholest} \rho_L(1,n) & \le \E \left( \sum_{|E_k| \le
E_{max}} \|\chi_n v_k \| \cdot \|\chi_1 v_k\| \right) \\
\nonumber & = \sum_{k = 1}^\infty \E \left( \chi_{\{ \omega :
|E_k(\omega)| \le E_{max} \} } \|\chi_n v_k \| \cdot \|\chi_1 v_k\| \right) \\
\nonumber & = \sum_{k = 1}^\infty \int_{\R^{2L}} \chi_{\{ \omega :
|E_k(\omega)| \le E_{max} \} } \|\chi_n v_k \| \cdot \|\chi_1
v_k\| \prod_{j = -L+1}^L r(\omega_j) \, d\omega_j
\end{align}

In the next subsection we will rewrite the integral in the last line of
\eqref{e.rholest} by introducing a change of variables based on
the Pr\"ufer phase of the eigenfunctions $v_k$ at the integer
sites in $[-L,L]$. The Jacobian of this change of variables will
be computed in Subsection~\ref{ss.jacobian}. This will then lead
to a formula for $\rho_L(1,n)$ involving integral operators, which
will be made explicit in Subsection~\ref{ss.intop}.

\subsection{Change of Variables} \label{sec:changevar}

In this subsection we introduce the change of variables according
to which we will rewrite the integrals appearing in the last line
of \eqref{e.rholest}.

To introduce Pr\"ufer variables, let $u_{-L}(\cdot,\omega,E)$ be the solution of $-u''+(W_0+V_{\omega})u=Eu$ satisfying $u(-L)=0$ and $u'(-L)=1$. The corresponding Pr\"ufer phase $\varphi_{-L}(\cdot,\omega,E)$ and amplitude $R_{-L}(\cdot,\omega,E)$ are defined by
$$
u_{-L} = R_{-L} \sin \varphi_{-L}, \quad u_{-L}' = R_{-L} \cos \varphi_{-L}
$$
normalized so that $\varphi_{-L}(-L)=0$ and $\varphi_{-L}(\cdot,\omega,E)$ continuous to get uniqueness of the phase. The Pr\"ufer phase satisfies the first order equation (cf.\ Lemma~\ref{lem:xder})
\begin{equation} \label{e.phaseDE}
\varphi_{-L}' = 1 - (1+W_0+V_{\omega}-E) \sin^2 \varphi_{-L}.
\end{equation}
Fix an $M>0$ such that
\begin{equation}\label{f.Mdef}
\mathrm{supp} \, r \subseteq [-M,M].
\end{equation}
Thus
\begin{equation} \label{e.shift1}
\|\varphi'\|_\infty \le 2 + \|W_0\|_{\infty} + M \|f\|_\infty +E_{max} < \infty.
\end{equation}
Choose $N \in \Z_+$ such that
\begin{equation} \label{e.shift2}
2 + \| W_0\|_\infty
+ M \|f\|_\infty + E_{max} < N \pi.
\end{equation}
Hence the change of the Pr\"ufer phase over any interval of length
one is uniformly bounded in absolute value by $N
\pi$.\footnote{Actually, we need $N \pi$ to bound the
\emph{growth} of the Pr\"ufer phase over a unit interval, but it
follows from the differential equation that it can never decrease
by more than $\pi$.}

With the circle
$$
\T_N := \R / (2 \pi N \Z)
$$
and
$$
\Omega :=\left\{ (\omega,k) \in [-M,M]^{2L} \times \Z_+ \,:\; E_k(\omega) \in [-E_{max},E_{max}] \right\}
$$
we can now define our change of variables
\begin{eqnarray*}
\mathcal{C} \quad : \quad  \Omega \quad
& \longrightarrow & \T_N^{2L-1} \times
[-E_{max} , E_{max}] \times \{ 0 , \ldots , 2N-1 \} \\
(\omega, k) & \mapsto & (\theta_{-L+1}
, \ldots , \theta_{L-1} , E , j )
\end{eqnarray*}
as follows:
\begin{itemize}

\item For $i = -L + 1, \ldots, L-1$, $\theta_i \in \T_N$ is chosen so that
$$
\varphi_{-L}(i,\omega,E_k(\omega)) \equiv \theta_i \mod 2\pi N.
$$

\item $E \in [-E_{max} , E_{max}]$ is given by
$$
E = E_k(\omega).
$$

\item Finally, $j \in \{ 0 , \ldots , 2N-1 \}$ is defined so that
$$
k \equiv j \mod 2N.
$$

\end{itemize}

\begin{lemma}\label{l.cis11}
The change of variables $\mathcal{C}$ is one-to-one.
\end{lemma}

\begin{proof}
The key point is that the Pr\"ufer phase of the $k$-th
eigenfunction of $H_\omega$ runs from $0$ (at $-L$) to $k \pi$ (at
$L$), i.e.\ $\varphi_{-L}(L,\omega,E_k(\omega)) =k\pi$. Since we are taking phases modulo $2N\pi$, we need to ensure
that no ambiguities are generated. The desired uniqueness follows
from our choice of $N$, which can be seen as follows.

Suppose that $(\theta_{-L+1}, \ldots, \theta_{L-1}, E, j)$ belongs to the range of $\mathcal{C}$ and that $\mathcal{C}(\omega,k) = (\theta_{-L+1}, \ldots, \theta_{L-1}, E, j)$, $(\omega,k) \in \Omega$. Let $\theta_{-L} :=0$, $\theta_{L} :=j\pi$. Thus, by definition of $\mathcal{C}$, for $i=-L+1,\ldots, L$,
\begin{equation} \label{e.Cinv}
\varphi_{-L}(i,\omega,E) \equiv \theta_i \mod 2\pi N.
\end{equation}
This determines $\omega_i$, $i=-L+1,\ldots,L$, uniquely, as is seen iteratively in $i$: For fixed $i$, $\varphi_{-L}(i-1,\omega,E)$ is determined by $\omega_{-L+1}$, \ldots, $\omega_{i-1}$ and, by \eqref{e.shift1} and \eqref{e.shift2},
$$
|\varphi_{-L}(i,\omega,E) - \varphi_{-L}(i-1,\omega,E)| < N\pi.
$$
We also know from Appendix~\ref{s.appB} (cf.\ Lemma~\ref{lem:lamderphi}) that $\varphi_{-L}(i,\omega,E)$ is strictly decreasing in $\omega_i$. Therefore, given $\omega_{-L+1}, \ldots, \omega_{i-1}$, there can be at most one $\omega_i \in [-M,M]$ satisfying \eqref{e.Cinv}.

Finally, with the unique values of $\omega_{-L+1}$, \ldots, $\omega_L$ reconstructed, $k \in \Z_+$ is uniquely determined by $\varphi_{-L}(L,\omega,E)=k\pi$.
\end{proof}

Now we carry out the change of variables in \eqref{e.rholest} and consider the
resulting integral.

Similarly to the definition of $u_{-L}$, let $u_L(\cdot,\omega,E)$ be the solution of $-u''+(W_0+V_{\omega})u=Eu$ determined by $u_L(L)=0$, $u'_L(L)=1$, with corresponding Pr\"ufer variables $\varphi_L$ and $R_L$, where $\varphi_L(L)=0$.

For $g\in L^{\infty}([-1,0])$ and $E, \lambda , \theta , \eta \in \R$, let $u_0(\cdot , \theta,
\lambda, g, E)$ be the unique solution of
\begin{equation}\label{geve}
-u'' + gu + \lambda fu = Eu
\end{equation}
with $u_0'(0) = \cos \theta$, $u_0(0) = \sin \theta$ and let
$u_{-1}(\cdot , \eta, \lambda, g, E)$ be the unique solution of
\eqref{geve} with $u_{-1}'(-1) = \cos \eta$, $u_{-1}(-1) = \sin
\eta$.

Let $\varphi_0(\cdot , \theta , \lambda , g, E)$ and $R_0(\cdot ,
\theta , \lambda , g, E)$ be the Pr\"ufer phase and amplitude,
respectively, for $u_0(\cdot , \theta , \lambda , g, E)$.
Similarly, let
$\varphi_{-1}(\cdot , \eta , \lambda , g, E)$ and
$R_{-1}(\cdot , \eta , \lambda , g, E)$ be the Pr\"ufer variables for $u_{-1}(\cdot , \eta,
\lambda, g, E)$.

For the next definition, in order to make use of \eqref{e.shift1} and \eqref{e.shift2}, we assume $\|g\|_{\infty} \le \|W_0\|_{\infty}$ and $|E|\le E_{max}$.

Note that $\varphi_0(\cdot , \theta+\pi , \lambda , g, E) = \varphi_0(\cdot , \theta , \lambda , g, E) + \pi$ and $\varphi_{-1}(\cdot , \eta+\pi , \lambda , g, E) = \varphi_{-1}(\cdot , \eta , \lambda , g, E)+\pi$. Thus, in particular, $\varphi_0(-1 , \cdot , \lambda , g, E)$ and $\varphi_{-1}(0 , \cdot , \lambda , g, E)$ induce well-defined mappings from $\T_N$ to $\T_N$, which is how we use them below.

If $\alpha, \beta \in \T_N$ are such that there exists a coupling
constant $\lambda \in [-M,M]$ with $\varphi_0(-1,\alpha,\lambda,g,E) = \beta$
(or, equivalently, $\varphi_{-1}(0,\beta,\lambda,g,E) = \alpha$), we
define $\lambda(\beta,\alpha,g,E) = \lambda$. Note that by \eqref{e.shift1} and \eqref{e.shift2}, using the same argument as in the proof of Lemma~\ref{l.cis11}, this $\lambda$ is uniquely determined if it exists.

Finally, write $f_i := f(\cdot -i)$ and $g_i := W_0(\cdot-i)$.

\begin{lemma}\label{l.aftercov}
We have
$$
\rho_L(1,n) \le \int_{-E_{max}}^{E_{max}} \rho_L(1,n,E) \, dE,
$$
where, for $E \in [-E_{max}, E_{max}]$, we write
\begin{align} \label{e.fixenergy}
\rho_L(1,n,E) := \sum_{j=0}^{2N-1} \int_{\T_N^{2L-1}} &
r(\lambda(\theta_{L-1}, j\pi, g_L, E)) \cdot r(\lambda(0,
\theta_{-L+1}, g_{-L+1}, E)) \\
& \left( \prod_{i - -L+2}^{L-1} r(\lambda(\theta_{i-1}, \theta_i,
g_i, E))  \right) \left( \int_{n-1}^n u_{L}^2 \right)^{1/2} \left(
\int_0^1 u_{L}^2 \right)^{1/2} \nonumber \\
& \frac{R_{-L}^2(-L+1) \cdots R_{-L}^2(0) \cdot R_{L}^2(1) \cdots
R_{L}^2(L-1)}{\int f_{-L+1} u_{-L}^2 \cdots \int f_{0} u_{-L}^2
\cdot \int f_{1} u_{L}^2 \cdots \int f_{L} u_{L}^2} \nonumber \\
& d \theta_{-L+1} \cdots d \theta_{L-1} \nonumber
\end{align}
and interpret $r(\lambda(\cdots))$ as zero if $\lambda(\cdots)$
does not exist. Here the argument $\omega$ in the functions $u_{\pm L} = u_{\pm L}(\cdot,\omega,E)$ and $R_{\pm L} = R_{\pm L}(\cdot,\omega,E)$ is the one uniquely determined via Lemma~\ref{l.cis11} by $\theta_{-L+1}$, \ldots, $\theta_{L}$, $E$ and $j$.
\end{lemma}

\begin{proof}
For $k \in \Z_+$, let
$$
\Omega_k := \{ \omega \in [-M,M]^{2L} : |E_k(\omega)| \le E_{max} \}
$$
and write
$$
A_k := \int_{\R^{2L}} \chi_{\Omega_k} \|\chi_n v_k\| \|\chi_1 v_k\| \prod_{i=-L+1}^L r(\omega_i)\,d\omega_i.
$$

On $\Omega_k$ we change variables by the map
\begin{eqnarray*}
\mathcal{C}_k \quad : \quad \Omega_k & \longrightarrow &
\T_N^{2L-1} \times
[-E_{max} , E_{max}] \\
\omega & \mapsto & (\theta_{-L+1} ,
\ldots , \theta_{L-1} , E_k(\omega)).
\end{eqnarray*}
Let $J_k = \partial \mathcal{C}_k/\partial \omega$ be its Jacobian. Pick $j\in \{0,\ldots,2N-1\}$ so that $j\equiv k \mod 2N$. Noting that, in terms of the new variables on $\mathcal{C}_k(\Omega_k)$,
\[ v_k = \frac{u_L(\cdot,\omega(\theta_{-L+1},\ldots,\theta_L,E,j),E)}{\|u_L(\cdot,\omega(\theta_{-L+1},\ldots,\theta_L,E,j),E)\|},\]
we get
\begin{eqnarray*}
A_k & = & \int_{-E_{max}}^{E_{max}} \int_{\T_N^{2L-1}} \chi_{\mathcal{C}_k(\Omega_k)} |\det J_k|^{-1} r(\lambda(0,\theta_{-L+1},g_{-L+1},E)) \\
& & \left( \prod_{i=-L+2}^{L-1} r(\lambda(\theta_{i-1}, \theta_i, g_i, E)) \right) r(\lambda(\theta_{L-1},j\pi, g_L,E)) \\
& & \frac{\left(\int_{n-1}^n u_L^2\right)^{1/2} \left(\int_0^1 u_L^2 \right)^{1/2}}{\int_{-L}^L u_L^2} \,d\theta_{-L+1} \ldots d\theta_{L-1}\,dE.
\end{eqnarray*}

The Jacobian determinant is calculated in Lemma~\ref{l.jacobidet} of the next subsection. Inserting the result yields
\begin{eqnarray*}
A_k & = & \int_{-E_{max}}^{E_{max}} \int_{\T_N^{2L-1}} \chi_{\mathcal{C}_k(\Omega_k)} r(\lambda(\theta_{L-1},j\pi, g_L,E)) r(\lambda(0,\theta_{-L+1},g_{-L+1},E)) \\
& & \left( \prod_{i=-L+2}^{L-1} r(\lambda(\theta_{i-1}, \theta_i, g_i, E)) \right) \left(\int_{n-1}^n u_L^2\right)^{1/2} \left(\int_0^1 u_L^2 \right)^{1/2} \\
& & \frac{R_{-L}^2(-L+1) \ldots R_{-L}^2(0) R_L^2(1) \ldots R_L^2(L-1)}{\int f_{-L+1} u_{-L}^2 \ldots \int f_0 u_{-L}^2 \int f_1 u_L^2 \ldots \int f_L u_L^2} \, d\theta_{-L+1} \ldots d\theta_{L-1}\,dE.
\end{eqnarray*}
Let $k_1 \not= k_2$ and $j_1$, $j_2 \in \{0,\ldots,2N-1\}$ with $j_{\ell} \equiv k_{\ell} \mod 2N$. Then Lemma~\ref{l.cis11} says that
$$
\left( \mathcal{C}_{k_1}(\Omega_{k_1}) \times \{j_1\} \right) \cap \left( \mathcal{C}_{k_2}(\Omega_{k_2}) \times \{j_2\} \right) = \emptyset.
$$
Thus it follows that $\sum_k A_k \le \int_{-E_{max}}^{E_{max}} \rho_L(1,n,E)$, with $\rho_L(1,n,E)$ defined in \eqref{e.fixenergy}. But by \eqref{e.rholest} we have $\rho_L(1,n) \le \sum_k A_k$, which completes the proof.
\end{proof}

\subsection{Calculation of the Jacobian}\label{ss.jacobian}

It will turn out that the Jacobians arising above, up to constant row and column multipliers, have the simple structure considered in the next lemma.

\begin{lemma}\label{l.detformula}
We have
$$
\det \begin{pmatrix} a_1 & a_1 & a_1 & \cdots & a_1 & a_1
\\ b_2 & a_2 & a_2 & \cdots & a_2 & a_2 \\
b_3 & b_3 & a_3 & \cdots & a_3 & a_3 \\
\vdots & \vdots & \vdots & & \vdots & \vdots\\
b_n & b_n & b_n & \cdots & b_n & a_n \end{pmatrix} = a_1 (a_2 -
b_2)(a_3 - b_3) \cdots (a_n - b_n).
$$
\end{lemma}

\begin{proof}
Observe that
\begin{align*}
\det \begin{pmatrix} a_1 & a_1 & a_1 & \cdots & a_1 & a_1
\\ b_2 & a_2 & a_2 & \cdots & a_2 & a_2 \\
b_3 & b_3 & a_3 & \cdots & a_3 & a_3 \\
\vdots & \vdots & \vdots & & \vdots & \vdots\\
b_n & b_n & b_n & \cdots & b_n & a_n \end{pmatrix} & = \det
\begin{pmatrix} a_1 & a_1 & a_1 & \cdots & a_1 & a_1
\\ b_2 - a_2 & 0 & 0 & \cdots & 0 & 0 \\
b_3 & b_3 & a_3 & \cdots & a_3 & a_3 \\
\vdots & \vdots & \vdots & & \vdots & \vdots\\
b_n & b_n & b_n & \cdots & b_n & a_n \end{pmatrix} \\
& = (a_2 - b_2) \det \begin{pmatrix} a_1 & a_1 & a_1 & \cdots &
a_1 & a_1
\\ b_3 & a_3 & a_3 & \cdots & a_3 & a_3 \\
b_4 & b_4 & a_4 & \cdots & a_4 & a_4 \\
\vdots & \vdots & \vdots & & \vdots & \vdots\\
b_n & b_n & b_n & \cdots & b_n & a_n \end{pmatrix}
\end{align*}
and then obtain the result by iteration (or induction).
\end{proof}

\begin{lemma}\label{l.jacobidet}
With the conventions for the arguments of $u_{\pm L}$ and $R_{\pm L}$ made in Lemma~\ref{l.aftercov} we have
$$
\det J_k = \frac{\int f_{-L+1} u_{-L}^2 \cdots \int f_{0} u_{-L}^2
\cdot \int f_{1} u_{L}^2 \cdots \int f_{L} u_{L}^2}{R_{-L}^2(-L+1)
\cdots R_{-L}^2(0) \cdot R_{L}^2(1) \cdots R_{L}^2(L-1)} \left(
\int_{-L}^L u_L^2 \right)^{-1}.
$$
\end{lemma}

\begin{proof}
With a slight adjustment of the notation introduced above, we have for $i = -L+1 , \ldots ,
0$,
\begin{equation} \label{e.isonleft}
\theta_i \equiv \varphi_{-L}(i,0,(\omega_{-L+1} , \ldots ,
\omega_i),E_k(\omega)) \mod 2\pi N,
\end{equation}
and for $i = 1 , \ldots , L-1$,
\begin{equation} \label{e.isonright}
\theta_i \equiv \varphi_{L}(i, j\pi, (\omega_{i+1} , \ldots ,
\omega_L),E_k(\omega)) \mod 2\pi N.
\end{equation}
The notational adjustment made here consists in stressing that $\varphi_{-L}(i)$ depends explicitly on $\omega_n$ only for $n=-L+1,\ldots,i$ (we only need to know the potential on $[-L,i]$ to calculate it), while it depends on all $\omega_n$ implicitly through $E_k(\omega)$. Similar reasoning applies to $\varphi_L(i)$ in \eqref{e.isonright}.

Thus, using \eqref{e.isonright} for $1 \le i \le L-1$ and $n \le i$, we have by Corollary~\ref{cor:Ederphi} and the Feynman-Hellmann formula
\begin{align*}
\frac{\partial \theta_i}{\partial \omega_n} & = \frac{\partial
\varphi_L(i)}{\partial E} \cdot \frac{\partial E_k}{\partial
\omega_n} \\
& = - \frac{1}{R_L^2(i)} \int_i^L u_L^2 \cdot \int f_n v_k^2 \\
& = - \frac{1}{R_L^2(i)} \frac{\int_i^L u_L^2}{\int_{-L}^L u_L^2}
\cdot \int f_n u_L^2;
\end{align*}
while for $1 \le i \le L-1$ and $n > i$, we have
\begin{align*}
\frac{\partial \theta_i}{\partial \omega_n} & = \frac{\partial
\varphi_L(i)}{\partial \omega_n} + \frac{\partial
\varphi_L(i)}{\partial E} \cdot \frac{\partial E_k}{\partial
\omega_n} \\
& = \frac{1}{R_L^2(i)} \int f_n u_L^2 - \frac{1}{R_L^2(i)} \int_i^L u_L^2 \cdot \int f_n v_k^2 \\
& = \frac{1}{R_L^2(i)} \int f_n u_L^2 - \frac{1}{R_L^2(i)}
\frac{\int_i^L u_L^2}{\int_{-L}^L u_L^2} \cdot \int f_n u_L^2 \\
& = \frac{1}{R_L^2(i)} \frac{\int_{-L}^i u_L^2}{\int_{-L}^L u_L^2}
\cdot \int f_n u_L^2.
\end{align*}
Analogous formulae, based on \eqref{e.isonleft}, hold in the case $ -L+1 \le i \le 0$. In this case we get for $n>i$ that
$$
\frac{\partial \theta_i}{\partial \omega_n} = \frac{1}{R^2_{-L}(i)} \frac{\int_{-L}^i u_{-L}^2}{\int_{-L}^L u_{-L}^2} \int f_n u_{-L}^2
$$
and, for $n\le i$,
$$
\frac{\partial \theta_i}{\partial \omega_n} = - \frac{1}{R_{-L}^2(i)} \frac{\int_i^L u_{-L}^2}{\int_{-L}^L u_{-L}^2} \int f_n u_{-L}^2.
$$

Also, writing $E$ as the first of the new variables, the first row of the Jacobian has entries $\partial E/\partial \omega_n = \int f_n v_k^2 = \int f_n u_L^2 / \int_{-L}^L u_L^2$.

Using these formulae and factoring out common factors in rows and
columns, we find that
$$
\det J_k = \frac{\int f_{-L+1} u_L^2 \cdots \int f_L
u_L^2}{R_{-L}^2(-L+1) \cdots R_{-L}^2(0) \cdot R_L^2(1) \cdots
R_L^2(L-1)} \left( \int_{-L}^L u_L^2 \right)^{-2L} \cdot \det A
$$
where
$$
A = \begin{pmatrix} 1 & 1 & 1 & \cdots & 1 & 1 & 1
\\ -I_{-L+1 , L}^- & I_{-L,-L+1}^- & I_{-L,-L+1}^- & \cdots & \cdots & I_{-L,-L+1}^- & I_{-L,-L+1}^- \\
-I_{-L+2 , L}^- & -I_{-L+2 , L}^- & I_{-L , -L+2}^- & \cdots & \cdots & I_{-L , -L+2}^- & I_{-L , -L+2}^- \\
\vdots & \vdots & \vdots & & \vdots & \vdots & \vdots \\
-I_{0 , L}^- & \cdots & -I_{0 , L}^- & I_{-L , 0}^- & I_{-L ,
0}^- & \cdots & I_{-L , 0}^- \\
-I_{1 , L}^+ & \cdots & -I_{1 , L}^+ & -I_{1 ,
L}^+ & I_{-L , 1}^+ & \cdots & I_{-L , 1}^+ \\
\vdots & \vdots & \vdots & & \vdots & \vdots & \vdots \\
-I_{L-1 , L}^+ & -I_{L-1 , L}^+ & -I_{L-1 , L}^+ & \cdots &
-I_{L-1 , L}^+ & -I_{L-1 , L}^+ & I_{-L , L-1}^+
\end{pmatrix}
$$
and
$$
I_{m,l}^\pm := \int_m^l u_{\pm L}^2.
$$
Applying Lemma~\ref{l.detformula}, we obtain
\begin{eqnarray*}
\det A & = & (I_{-L,-L+1}^- + I_{-L+1 , L}^-) \cdots (I_{-L ,0}^-
+ I_{0 , L}^-) \\ & & \cdot (I_{-L , 1}^+ + I_{1 , L}^+) \cdots (I_{-L ,
L-1}^+ + I_{L-1 , L}^+) \\
& = & \left( \int_{-L}^L u_{-L}^2 \right)^{L} \cdot \left(
\int_{-L}^L u_L^2 \right)^{L-1}.
\end{eqnarray*}
Plugging this into the formula for $\det J_k$ obtained above, the
result follows.
\end{proof}

\subsection{The Integral Operator Formula}\label{ss.intop}

The expression in Lemma~\ref{l.aftercov} may be written in a more
succinct form once we have introduced a number of quantities.

If, for given $g \in L^\infty(-1,0)$ with $\|g\|_{\infty} \le \|W_0\|_{\infty}$, $E \in [-E_{max} , E_{max}]$ and $\beta, \alpha \in \T_N$, $\lambda(\beta,\alpha, g,E)$ as defined in Section~\ref{sec:changevar} exists, we
write
\begin{eqnarray} \label{upmdef}
u_+(\cdot, \beta, \alpha, g,E) & := & u_0(\cdot , \alpha, \lambda(\beta,
\alpha, g,E), g,E), \\  u_-(\cdot, \beta, \alpha, g,E) & := &
u_{-1}(\cdot , \beta, \lambda(\beta, \alpha, g,E), g,E) \nonumber
\end{eqnarray}
and
\begin{eqnarray} \label{Rpmdef}
R_+(\cdot, \beta, \alpha, g,E) & := & R_0(\cdot , \alpha, \lambda(\beta,
\alpha, g,E), g,E), \\  R_-(\cdot, \beta, \alpha, g,E) & := &
R_{-1}(\cdot , \beta, \lambda(\beta, \alpha, g,E), g,E). \nonumber
\end{eqnarray}

For later use, note that
\begin{align}\label{upmrel}
u_+^2 (\cdot, \beta , \alpha,g,E) & = R_+^2 (-1,\beta,\alpha,g,E) \cdot
u_-^2(\cdot,\beta,\alpha,g,E) \\
\nonumber & = R_-^{-2} (0,\beta,\alpha,g,E) \cdot
u_-^2(\cdot,\beta,\alpha,g,E).
\end{align}

We introduce the following integral operators defined on functions $F$ on $\T_N$:
\begin{align*}
(T_0(g,E) F)(\beta) & = \int \frac{R_+^2(-1,\beta,\alpha,g,E)}{\int
f u_+^2(\cdot,\beta,\alpha,g,E)} r(\lambda(\beta,\alpha,g,E))
F(\alpha) \, d\alpha, \\
(\tilde T_0(g,E) F)(\alpha) & = \int
\frac{R_-^2(0,\beta,\alpha,g,E)}{\int f
u_-^2(\cdot,\beta,\alpha,g,E)} r(\lambda(\beta,\alpha,g,E))
F(\beta) \, d\beta, \\
(T_1(g,E) F)(\beta) & = \int \frac{R_+(-1,\beta,\alpha,g,E)}{\int f
u_+^2(\cdot,\beta,\alpha,g,E)} r(\lambda(\beta,\alpha,g,E))
F(\alpha) \, d\alpha,
\end{align*}
and, for $j = 0, \ldots, 2N-1$, the functions
\begin{align*}
(\Psi_j(g,E))(\theta) & = \frac{R_+^2(-1,\theta,j\pi,g,E)}{\int f
u_+^2(\cdot,\theta,j\pi,g,E)} r(\lambda(\theta,j\pi,g,E)), \\
(\Phi(g,E))(\theta) & = \frac{R_-^2(0,0,\theta,g,E)}{\int
f u_-^2(\cdot,0,\theta,g,E)} r(\lambda(0,\theta,g,E)).
\end{align*}

Now we are finally in a position to state the integral operator
formula, which bounds $\rho_L(1,n,E)$ from above.

\begin{lemma}
There exists a constant $C = C(E_{max})$ such that for every $E
\in [-E_{max} , E_{max}]$, we have
\begin{eqnarray} \label{e.intopformula}
\lefteqn{\rho_L(1,n,E)} \\ & \le & C \sum_{j = 0}^{2N-1} \Big\langle \tilde T_0(g_0,E)
\cdots \tilde T_0(g_{-L+2},E) \Phi(g_{-L+1},E) , \nonumber \\
& &  T_1(g_1,E) \cdots
T_1(g_n,E) T_0(g_{n+1},E) \cdots T_0(g_{L-1},E) \Psi_j(g_L,E) \Big\rangle. \nonumber
\end{eqnarray}
Here $\langle \cdot, \cdot \rangle$ denotes the inner product on $L^2(\T_N)$.

\end{lemma}

\begin{proof}
It follows from the a priori bounds in Lemma~\ref{lem:solest} that there exists a constant $C=C(E_{max})$ such that
$$
\int_{n-1}^n u_L^2 \le C R_L^2(n), \quad \int_0^1 u_L^2 \le C R_L^2(0)
$$
uniformly in $E\in [-E_{max},E_{max}]$. After using these bounds in the integrand on the right hand side of \eqref{e.fixenergy} we rearrange the terms in the integrand as
\begin{eqnarray*}
\lefteqn{ R_L(n) R_L(0) \frac{R_{-L}^2(-L+1) \cdots R_{-L}^2(0) R_L^2(1) \cdots R_L^2(L-1)}{ \int f_{-L+1}  u_{-L}^2 \cdots \int f_0 u_{-L}^2 \int f_1 u_L^2 \cdots \int f_L u_L^2}} \\
& = & \left( \prod_{i=-L+1}^0 \frac{R_{-L}^2(i)}{\int f_i u_{-L}^2} \right) \left( \prod_{i=1}^n \frac{R_L(i-1)R_L(i)}{\int f_i u_L^2} \right) \left( \prod_{i=n+1}^L \frac{R_L^2(i-1)}{\int f_i u_L^2} \right).\nonumber
\end{eqnarray*}
Taking into account the scaling properties of Pr\"ufer amplitudes, we get the following relations between $u_{\pm L}$, $R_{\pm L}$ and $u_{\pm}$, $R_{\pm}$:
$$
\frac{R_{-L}^2(i)}{\int f_i u_{-L}^2} = \frac{R_-^2(0,\theta_{i-1},\theta_i,g_i,E)}{\int f u_-^2(\cdot,\theta_{i-1}, \theta_i,g_i,E)},
$$
$$
\frac{R_L^2(i-1)}{\int f_i u_L^2} = \frac{R_+^2(-1, \theta_{i-1}, \theta_i, g_i,E)}{\int f u_+^2(\cdot, \theta_{i-1}, \theta_i, g_i, E)},
$$
and
\begin{eqnarray*}
\frac{R_L(i-1) R_L(i)}{\int f_i u_L^2} & = & \frac{R_L^2(i-1)}{\int f_i u_L^2} \cdot \frac{R_L(i)}{R_L(i-1)} \\
& = & \frac{R_+^2(-1, \theta_{i-1}, \theta_i, g_i, E)}{\int f u_+^2 (\cdot, \theta_{i-1}, \theta_i, g_i, E)} \cdot \frac{1}{R_+(-1,\theta_{i-1}, \theta_i, g_i, E)} \\
& = & \frac{R_+(-1, \theta_{i-1}, \theta_i, g_i, E)}{\int f u_+^2(\cdot, \theta_{i-1}, \theta_i, g_i, E)}.
\end{eqnarray*}
Plugging all this into \eqref{e.fixenergy} and using the definitions of the integral operators $T_0$, $\tilde{T}_0$, $T_1$ as well as the functions $\Psi_j$ and $\Phi$, we obtain \eqref{e.intopformula}.
\end{proof}

\subsection{Proof of Proposition~\ref{l.finitevolkeylemma} and Theorem~\ref{main}} \label{ss.endofproof}

We are now in a position to describe how our main result Theorem~\ref{main} follows from norm bounds for the operators $T_1$, $T_0$ and $\tilde{T}_0$, which we will establish in the remaining sections of this paper.

As was explained in Section~\ref{ss:fvcorrelators}, it suffices to prove Proposition~\ref{l.finitevolkeylemma}. By Lemma~\ref{l.aftercov} it suffices to establish a bound
\begin{equation} \label{e.rho1nEbound}
\rho_L(1,n,E) \le Ce^{-\eta n}
\end{equation}
with constants $C<\infty$ and $\eta>0$ which are uniform in $E\in [-E_{max},E_{max}]$. For this we will use the integral formula \eqref{e.intopformula}.

Denote the norm of a linear operator $T$ from $L^p(\T_N)$ to
$L^q(\T_N)$ by $\| T \|_{p,q}$.

By Lemmas~\ref{l.T011bound} and \ref{l.T012bound} we have $\|T_0(g,E)\|_{1,1}=1$ as well as $\|T_0(g,E)\|_{1,2} \le C$ and $\|\Psi_j(g,E)\|_1 \le C$ uniformly in $E\in [-E_{max},E_{max}]$, $\|g\|_{\infty} \le \|W\|_{\infty}$ and $j=0,\ldots,2N-1$. Thus
$$
\| T_0(g_{n+1},E) \cdots T_0(g_{L-1},E) \Psi_j(g_L,E)\|_2 \le C
$$
uniformly in $E\in [-E_{max},E_{max}]$, $L\in \N$ and $j=0,\ldots,2N-1$. Similarly, also uniformly,
$$
\|\tilde{T}_0(g_0,E) \cdots \tilde{T}(g_{-L+2},E) \Phi(g_{-L+1},E)\|_2 \le C.
$$
This yields, by \eqref{e.intopformula} and Cauchy-Schwarz, that there is $C=C(E_{max},W_0,N)$ such that
$$
\rho_L(1,n,E) \le C \prod_{i=1}^n \|T_1(g_i,E)\|_{2,2}.
$$
In Section~\ref{s.T122strongbound} we will show that $\|T_1(g,E)\|_{2,2} < 1$ for every $g\in L^{\infty}(-1,0)$ and $E\in \R$. Finally, we establish in Section~\ref{s.T122continuity} that $\|T_1(g,E)\|_{2,2} = \|T_1(g-E,0)\|_{2,2}$ is continuous in $(g,E) \in L^{\infty}(-1,0) \times \R$. By assumption \eqref{relcomp} it is guaranteed that $\{g_i:i\in \Z\}$ is relatively compact in $L^{\infty}(-1,0)$ and thus $\{g_i:i\in \Z\} \times [-E_{max},E_{max}]$ is relatively compact in $L^{\infty}(-1,0)\times \R$. Thus
$$
\|T_1(g_i,E)\|_{2,2} \le \gamma < 1
$$
uniformly in $i$ and $E\in [-E_{max},E_{max}]$. This proves \eqref{e.rho1nEbound} with $\eta = \ln(1/\gamma)$.

\section{Elementary Results for the Integral Operators}

In this section we consider the integral operators introduced in Section~\ref{ss.intop} and establish several elementary results for them.

We begin with $T_0$ and $\tilde T_0$.

\begin{lemma} \label{l.T011bound}
We have
$$
\|T_0(g,E)\|_{1,1} = \|\tilde T_0(g,E)\|_{1,1} = 1.
$$
\end{lemma}

\begin{proof}
As $g$, $E$ are fixed here, we will suppress them in this proof.
Suppose $F \in L^1(\T_N)$. Then,
\begin{eqnarray} \label{e.t0bound}
\|T_0 F\|_1 & \le & \int_{\T_N} \int_{\T_N} \frac{R_+^2(-1,\beta,\alpha)}{\int f
u_+^2(\cdot,\beta,\alpha)} r(\lambda(\beta,\alpha)) |F(\alpha)|
\, d\alpha \, d\beta \\
& = & \int_{\T_N} \int_{\T_N} \frac{R_+^2(-1,\beta,\alpha)}{\int f
u_+^2(\cdot,\beta,\alpha)} r(\lambda(\beta,\alpha))
\, d\beta \,  |F(\alpha)| \, d\alpha. \nonumber
\end{eqnarray}
For fixed $\alpha$, $\lambda(\beta,\alpha)$ is strictly increasing in $\beta$ and the inverse function satisfies, see Lemma~\ref{lem:lamderphi}(b),
\begin{equation} \label{e.betalambda}
\frac{\partial \beta}{\partial \lambda} = - R_+^{-2}(-1,\beta,\alpha) \int_0^{-1} f u_+^2(\cdot,\beta,\alpha) =R_+^{-2}(-1,\beta,\alpha) \int_{-1}^0 f u_+^2(\cdot,\beta,\alpha).
\end{equation}
Thus we can change variables and find that the right hand side of \eqref{e.t0bound} is equal to
$$
\int_{\T_N} \int_{\R} r(\lambda) \, d\lambda \, |F(\alpha)| \, d\alpha = \|r\|_1 \|F\|_1 = \|F\|_1,
$$
where we also used the fact
that $r$ is the density of a probability distribution. This shows that $\|T_0\|_{1,1} \le 1$. Since the first
step in \eqref{e.t0bound} becomes an identity when $F \ge 0$, we get
$\|T_0\|_{1,1} = 1$.

Using
\begin{equation} \label{e.alphalambda}
\frac{d\alpha}{d\lambda} = - R_-^{-2}(0,\beta,\alpha) \int_{-1}^0 f
u_-^2(\cdot,\beta,\alpha).
\end{equation}
instead of \eqref{e.betalambda}, the proof of $\|\tilde T_0\|_{1,1} = 1$
is completely analogous.
\end{proof}

\begin{lemma} \label{l.T012bound}
We have
$$
\|T_0(g,E)\|_{1,2} \le C <\infty
$$
and
$$
\|\tilde T_0(g,E)\|_{1,2} \le C <\infty
$$
uniformly in $E\in [-E_{max},E_{max}]$ and $\|g\|_{\infty} \le \|W_0\|_{\infty}$.

Denoting by $\|\cdot\|_1$ the $L^1$-norm on $\T_N$, we also have
$$
\|\Psi_j(g,E)\|_1 \le C<\infty
$$
and
$$
\|\Phi(g,E)\|_1 \le C<\infty
$$
uniformly in $E\in [-E_{max},E_{max}]$, $\|g\|_{\infty} \le \|W_0\|_{\infty}$ and $j=0,\ldots, 2N-1$.
\end{lemma}

\begin{proof}
Lemmas~\ref{lem:solest} and \ref{lem:L2b} provide bounds $C_1<\infty$ and $C_2>0$ such that
\begin{equation} \label{auxbound1}
R_+^2(-1,\beta,\alpha) = R_0^2(-1,\alpha,\lambda(\beta,\alpha)) \le C_1 \end{equation}
and
\begin{equation} \label{auxbound2}
\int f u_+^2(\cdot,\beta,\alpha) = \int f u_0^2(\cdot,\alpha,\lambda(\beta,\alpha)) \ge C_2
\end{equation}
uniformly in $E\in [-E_{max},E_{max}]$, $\|g\|_{\infty} \le \|W_0\|_{\infty}$ and $\alpha, \beta$ such that $\lambda(\beta,\alpha) \in \mbox{supp}\,r$. In \eqref{auxbound2} we have also exploited the assumption \eqref{singlesite1} on the single site potential~$f$.

Thus we get for $F\in L^1(\T_N)$ that
\begin{eqnarray*}
\|T_0 F\|_2^2 & = & \int_{\T_N} \left| \int_{\T_N} \frac{R_+^2(-1,\beta,\alpha)}{\int f u_+^2(\cdot,\beta,\alpha)} r(\lambda(\beta,\alpha)) F(\alpha)\,d\alpha \right|^2\,d\beta \\
& \le & 2\pi N (C_1/C_2)^2 \|r\|_{\infty}^2 \|F\|_1^2,
\end{eqnarray*}
resulting in the required norm bound for $T_0$. The bound for $\tilde{T}_0$ is found similarly.

From (\ref{auxbound1}) and (\ref{auxbound2}) we also get the bounds for $\Psi_j$ and $\Phi$, first in the $L^{\infty}$-norm and then in the $L^1$-norm since $\T_N$ has finite volume.
\end{proof}

Let us now turn to $T_1$. For $\alpha,\beta \in \T_N$, we write
\begin{equation}\label{t1kerdef}
T_1(\beta,\alpha) = \begin{cases} \frac{R_+(-1,\beta,\alpha)
r(\lambda(\beta,\alpha))}{\int f(x) u_+^2(x,\beta,\alpha) \, dx} &
\text{ if } \lambda(\beta,\alpha) \text{ exists,} \\
0 & \text{ otherwise}
\end{cases}
\end{equation}
for its integral kernel.

\begin{prop} \label{proposition1}
We have
\begin{align}
\label{t1kerper} & T_1(\beta+ \pi , \alpha + \pi) = T_1(\beta ,
\alpha) \text{ for all
} \beta,\alpha \in \T_N, \\
\label{t1kercont} & T_1(\cdot,\cdot) \text{ is continuous on }
\T_N^2.
\end{align}
\end{prop}

\begin{proof}
Note that $\varphi(x,\alpha+\pi,\lambda) =
\varphi(x,\alpha,\lambda) + \pi$ for every $x$. This follows from
$u_0(x,\alpha+\pi,\lambda) = - u_0(x,\alpha,\lambda)$ for every
$x$, along with the initial condition $\phi(0,\alpha,\lambda) =
\alpha$.

To derive \eqref{t1kerper} from this, consider a pair
$(\beta,\alpha)$ such that $\lambda(\beta,\alpha)$ exists. Then
$\lambda(\beta+\pi,\alpha+\pi)$ exists, too, and is equal to
$\lambda(\beta,\alpha)$. Moreover, it then follows readily from
the definition that $T_1(\beta+\pi,\alpha+\pi) =
T_1(\beta,\alpha)$.



To show \eqref{t1kercont}, we will need the following:
\begin{equation}\label{dopen}
D := \{ (\beta,\alpha) \in \T_N^2 : \lambda(\beta,\alpha) \text{
exists} \} \text{ is open and $\lambda (\cdot,\cdot)$ is
continuous on } D.
\end{equation}
To see this, fix some $(\beta,\alpha) \in D$ and $\varepsilon >
0$. Keep $\beta$ initially fixed and increase $\alpha$ a bit.
Clearly, there will still be a corresponding $\lambda$ that sends
$\beta$ to $\alpha + \delta_1$, $\delta_1 > 0$. Choose $\delta_1$
small enough so that $\lambda(\beta,\alpha + \delta_1) \le
\lambda(\beta,\alpha) + \frac{\varepsilon}{2}$. Similarly, keeping
$\alpha + \delta_1$ fixed and decreasing $\beta$ a bit, we find
$\delta_2
> 0$ so that $\lambda(\beta - \delta_2 , \alpha + \delta_1) \le
\lambda (\beta,\alpha) + \varepsilon$. Similarly, we can choose
$\delta_3,\delta_4
> 0$ with $\lambda(\beta + \delta_4,\alpha - \delta_3) \ge
\lambda(\beta, \alpha) - \varepsilon$. It then follows, again by
the monotonicity properties, that the set $\{ (\beta + \delta ,
\alpha + \tilde \delta) : - \delta_2 \le \delta \le \delta_4 , \,
- \delta_3 \le \tilde \delta \le \delta_1 \}$ is contained in $D$
and $\lambda$ restricted to this set takes values in the interval
$[\lambda(\beta, \alpha) - \varepsilon , \lambda(\beta, \alpha) +
\varepsilon]$. The assertion \eqref{dopen} follows.

With the closed subset $A := \lambda^{-1} (\mathrm{supp} \, r)$ of
$D$, we can rewrite $T_1(\beta,\alpha)$ as
$$
T_1(\beta,\alpha) = \begin{cases} \frac{R_+(-1,\beta,\alpha)
r(\lambda(\beta,\alpha))}{\int f(x) u_+^2(x,\beta,\alpha) \, dx} &
\text{ if } (\beta,\alpha) \in D, \\
0 & \text{ if } (\beta,\alpha) \in \T_N^2 \setminus A.
\end{cases}
$$
Notice that this is well-defined. Since $\{ D , \T_N^2 \setminus
A\}$ is an open cover of $\T_N^2$, it suffices to check continuity
for each of these two open sets. Continuity on $\T_N^2 \setminus A$
is obvious. Continuity on $D$ follows by \eqref{dopen}, the
continuity of $\lambda$ in $(\beta,\alpha)$, and (via \eqref{upmdef} and \eqref{Rpmdef}) the joint continuity of
$R_0(-1,\alpha,\lambda)$ and $\int f u_0^2(\cdot,\alpha,\lambda)$ in
$(\alpha,\lambda)$. The latter is a consequence of the a priori bound provided in Lemma~\ref{lem:contdep}. This concludes the proof of \eqref{t1kercont}.
\end{proof}

\begin{lemma} \label{l.T122elembound}
We have
$$
\|T_1(g,E)\|_{2,2} \le 1.
$$
\end{lemma}

\begin{proof}
Let
$$
K_1(\beta,\alpha) = \frac{r(\lambda(\beta,\alpha))}{\int f
u_+^2(\cdot,\beta,\alpha)}
$$
and
$$
K_2(\beta,\alpha) = \frac{R_+^2(-1,\beta,\alpha)
r(\lambda(\beta,\alpha))}{\int f u_+^2(\cdot,\beta,\alpha)}
$$
if $\lambda(\beta,\alpha)$ exists and $K_1(\beta,\alpha) =
K_2(\beta,\alpha) = 0$ otherwise.

Thus, using \eqref{upmrel} and the change of variables \eqref{e.alphalambda},
\begin{align}
\label{k1one} \int_{\T_N} K_1(\beta,\alpha) \, d\alpha & = \int_{\T_N}
\frac{r(\lambda(\beta,\alpha))}{\int f u_+^2(\cdot,\beta,\alpha)} \, d\alpha\\
\nonumber & = \int_{\T_N} \frac{R_-^2(0,\beta,\alpha) \cdot
r(\lambda(\beta,\alpha))}{\int f
u_-^2(\cdot,\beta,\alpha)} \, d\alpha \\
\nonumber & = \int_{\T_N} \frac{R_{-1}^2(0,\beta,\lambda) \cdot
r(\lambda)}{\int f u_{-1}^2(\cdot,\beta,\lambda)} \left|
\frac{d\alpha}{d\lambda} \right|
d\lambda \\
\nonumber & = \int_{\T_N} r(\lambda) \, d\lambda \\
\nonumber & = 1.
\end{align}

Similarly, using \eqref{e.betalambda},
\begin{align}
\label{k2one} \int_{\T_N} K_2(\beta,\alpha) \, d\beta & = \int_{\T_N}
\frac{R_+^2(-1,\beta,\alpha)
r(\lambda(\beta,\alpha))}{\int f u_+^2(\cdot,\beta,\alpha)} \, d\beta \\
\nonumber & = \int_{\T_N} \frac{R_0^2(-1,\alpha,\lambda) \cdot
r(\lambda)}{\int f u_0^2(\cdot,\alpha,\lambda)} \left|
\frac{d\beta}{d\lambda} \right|
d\lambda \\
\nonumber & = \int_{\T_N} r(\lambda) \, d\lambda \\
\nonumber & = 1.
\end{align}

We have $T_1(\beta,\alpha) = \sqrt{K_1(\beta,\alpha)}
\sqrt{K_2(\beta,\alpha)}$, so that the Schur Test, e.g.\ \cite{Weidmann}, immediately
gives $\|T_1\|_{2,2} \le 1$.
\end{proof}

\section{The Operator $T_1$ has $\|\cdot\|_{2,2}$-Norm Less Than One} \label{s.T122strongbound}

The purpose of this section is to establish to following strengthening of Lemma~\ref{l.T122elembound}, which is the key technical result of our work.

\begin{prop}\label{t1normthm}
We have $\|T_1(g,E)\|_{2,2} < 1$.
\end{prop}

We will suppress the $(g,E)$-dependence in our notation for the remainder of this section.

We have already seen that $T_1$ is a bounded operator on $L^2(\T_N)$. Moreover, \eqref{t1kerper} suggests that we decompose $T_1$ as a direct sum of integral operators on $L^2(0,\pi)$. Let us implement this:

\begin{lemma}\label{intdeclem}
{\rm(a)} Suppose $h$ is continuous on $(\pi n, \pi (n+1))$ for $n = 0,1,\ldots,2N-1$, $j \in \{ 0 , 1 , \ldots , 2N-1 \}$, and $x \in (0,\pi)$, and let
$$
(Uh)_j(x) = \frac{1}{\sqrt{2N}} \sum_{n = 0}^{2N-1} e^{\frac{-i \pi j n}{N}} h(x + \pi n).
$$
Then $U$ extends to a unitary operator
$$
U : L^2(\T_N) \to \bigoplus_{j=0}^{2N-1} L^2(0,\pi).
$$

{\rm (b)} We have
$$
U T_1 U^{-1} = \bigoplus_{j=0}^{2N-1} L_j,
$$
where $L_j$ is the integral operator in $L^2(0,\pi)$ with kernel
$$
L_j(\beta , \alpha) = \sum_{n = 0}^{2N-1} T_1(\beta , \alpha + n \pi) e^{\frac{i \pi j n}{N}}.
$$

{\rm (c)} We have $\|T_1\| = \|L_0\|$, with both norms being the operator norm in the respective $L^2$ space.
\end{lemma}

\begin{proof}
(a) Suppose $h$ is continuous on $(\pi n, \pi (n+1))$ for $n = 0,1,\ldots,2N-1$. Then,
\begin{align*}
\| Uh \|^2 & = \sum_{j = 0}^{2N-1} \int_0^\pi \left| \frac{1}{\sqrt{2N}} \sum_{n = 0}^{2N-1} e^{\frac{-i \pi j n}{N}} h(x + \pi n) \right|^2 \, dx \\
& = \int_0^\pi \sum_{j = 0}^{2N-1} \left| \sum_{n = 0}^{2N-1} \frac{e^{\frac{-i \pi j n}{N}}}{\sqrt{2N}}\, h(x + \pi n) \right|^2 \, dx \\
& = \int_0^\pi \sum_{j = 0}^{2N-1} |h(x + \pi j)|^2 \, dx \\
& = \int_0^{2N\pi} |h(x)|^2 \, dx \\
& = \|h\|^2.
\end{align*}
Here, all steps save the third follow by simple rewriting and the third step follows from the Parseval identity for $\C^{2N}$.

For a continuous $g = (g_j) \in \bigoplus_{j=0}^{2N-1} L^2(0,\pi)$, we define
$$
h(x + \pi n) = \frac{1}{\sqrt{2N}} \sum_{j = 0}^{2N-1} e^{\frac{i \pi j n}{N}} g_j(x),
$$
where $x \in (0,\pi)$ and $n \in \{ 0 , 1 , \ldots 2N-1 \}$ and note that $g = U h$. Since $h$ is continuous on $(\pi n, \pi (n+1))$ for $n = 0,1,\ldots,2N-1$, we may conclude that $U$ is a densely defined isometry with dense range, and hence $U$ extends to a unitary operator from $L^2(\T_N)$ onto $\bigoplus_{j=0}^{2N-1} L^2(0,\pi)$.

(b) We have
\begin{align*}
\Big( \bigoplus_{k=0}^{2N-1} L_k U h \Big)_j (\beta) & = \int_0^\pi L_j(\beta,\alpha) (Uh)_j(\alpha) \, d\alpha \\
& = \int_0^\pi \sum_{n = 0}^{2N-1} T_1(\beta , \alpha + n \pi) e^{\frac{i \pi j n}{N}} \frac{1}{\sqrt{2N}} \sum_{m = 0}^{2N-1} e^{\frac{-i \pi j m}{N}} h(\alpha + \pi m) \, d\alpha \\
& = \frac{1}{\sqrt{2N}} \sum_{n = 0}^{2N-1} \; \sum_{m = 0}^{2N-1} e^{\frac{-i \pi j (m-n)}{N}} \int_0^\pi T_1(\beta , \alpha + n \pi) h(\alpha + \pi m) \, d\alpha
\end{align*}
and
\begin{align*}
\Big( U T_1 h \Big)_j (\beta) & = \Big( U \int_0^{2N\pi} T_1(\cdot , \alpha) h (\alpha) \, d\alpha \Big)_j (\beta) \\
& = \frac{1}{\sqrt{2N}} \sum_{n = 0}^{2N-1} e^{\frac{-i \pi j n}{N}} \int_0^{2N\pi} T_1(\beta + \pi n , \alpha) h (\alpha) \, d\alpha \\
& = \frac{1}{\sqrt{2N}} \sum_{n = 0}^{2N-1} e^{\frac{-i \pi j n}{N}} \sum_{m=0}^{2N-1} \int_0^{\pi} T_1(\beta + \pi n , \alpha + \pi m) h (\alpha + \pi m) \, d\alpha \\
& = \frac{1}{\sqrt{2N}} \sum_{n = 0}^{2N-1} \; \sum_{m=0}^{2N-1} e^{\frac{-i \pi j n}{N}} \int_0^{\pi} T_1(\beta , \alpha + \pi(m-n)) h (\alpha + \pi m) \, d\alpha \\
& = \frac{1}{\sqrt{2N}} \sum_{\tilde n = 0}^{2N-1} \; \sum_{m=0}^{2N-1} e^{\frac{-i \pi j (m - \tilde n)}{N}} \int_0^{\pi} T_1(\beta , \alpha + \pi \tilde n) h (\alpha + \pi m) \, d\alpha,
\end{align*}
from which the asserted identity follows.

(c) From the decomposition established above, we get
$\|T_1\| = \max_{0 \le j \le 2N-1} \| L_j \|$. As $T_1(\beta,\alpha) \ge 0$, we have $|L_j(\beta,\alpha)|
\le L_0(\beta,\alpha)$ and therefore $\| L_j \| \le \| L_0 \|$ for every $j$. This yields the claim.
\end{proof}

\begin{proof}[Proof of Proposition~\ref{t1normthm}.]
By Lemma~\ref{intdeclem}.(c), it suffices to show $\|L_0\| < 1$. By
Lemma~\ref{l.T122elembound} and Lemma~\ref{intdeclem}.(c), $\|L_0\| \le 1$.
Suppose that $\|L_0\| = 1$. By
compactness, there exists $f \not= 0$ such that $\|L_0 f\| =
\|f\|$ (choose $f$ as an eigenvector to the eigenvalue $1 = \|L_0\|
= \| |L_0| \|$ of $| L_0 |$ and use $\|L_0 f\| = \| |L_0| f \|$).
The operator $L_0$ has a positive kernel and we may therefore
assume that $f \ge 0$.

Consider the $\pi$-periodic extension $\tilde f$ of $f$ to $\T_N$. Then,
\begin{align*}
(L_0 f)(\beta) & = \int_0^\pi \sum_{n = 0}^{2N-1} T_1(\beta , \alpha + n \pi) f(\alpha) \, d\alpha\\
& = \sum_{n = 0}^{2N-1} \int_0^\pi T_1(\beta , \alpha + n \pi) \tilde f(\alpha + n \pi) \, d\alpha\\
& = \int_{\T_N} T_1(\beta,\alpha) \tilde f(\alpha) \,
d\alpha.
\end{align*}
By \eqref{k1one}, \eqref{k2one} along with $K_1(\beta + \pi,\alpha
+ \pi) = K_1(\beta,\alpha)$ and $K_2(\beta + \pi,\alpha + \pi) =
K_2(\beta,\alpha)$, we find
\begin{align*}
\| L_0 f \|^2 & = \int_0^\pi | (L_0 f)(\beta) |^2 \, d\beta \\
& = \int_0^\pi \left| \int_{\T_N} T_1(\beta,\alpha) \tilde f(\alpha) \,
d\alpha \right|^2 \, d\beta \\
& = \int_0^\pi \left| \int_{\T_N}
\sqrt{K_1(\beta,\alpha)} \sqrt{K_2(\beta,\alpha)} \tilde f(\alpha)
\, d\alpha \right|^2 \, d\beta \\
& \le \int_0^\pi \left( \int_{\T_N} K_1(\beta,\alpha) \, d\alpha \;
\int_{\T_N} K_2(\beta,\alpha) |\tilde f(\alpha)|^2 \, d\alpha \right)
\, d\beta \\
& = \int_0^\pi \int_{\T_N} K_2(\beta,\alpha) |\tilde f(\alpha)|^2 \, d\alpha
\, d\beta \\
& = \int_0^\pi \sum_{n=0}^{2N-1} \int_0^\pi K_2(\beta,\alpha - \pi n) |\tilde f(\alpha - \pi n)|^2 \,
d\alpha \, d\beta \\
& = \int_0^\pi \sum_{n=0}^{2N-1} \int_0^\pi K_2(\beta + \pi n,\alpha) |\tilde f(\alpha)|^2 \,
d\alpha \, d\beta \\
& = \int_0^\pi \left( \sum_{n=0}^{2N-1} \int_0^\pi K_2(\beta + \pi n,\alpha) \, d\beta \right) |\tilde f(\alpha)|^2 \,
d\alpha \\
& = \int_0^\pi \left( \int_{\T_N} K_2(\beta,\alpha) \, d\beta \right) |\tilde f(\alpha)|^2 \,
d\alpha \\
& = \int_0^\pi |\tilde f(\alpha)|^2 \, d\alpha \\
& = \int_0^\pi |f(\alpha)|^2 \, d\alpha \\
& = \| f \|^2 \\
& = \| L_0 f \|^2.
\end{align*}
Thus, we have equality in the application of the Cauchy-Schwarz
inequality, which implies that for almost every $\beta \in
(0,\pi)$, the functions $\sqrt{K_1(\beta,\cdot)}$ and
$\sqrt{K_2(\beta,\cdot)} \tilde f(\cdot)$ are linearly dependent
in $L^2(\T_N)$. Since they are both non-negative and non-zero, we
see that for $\beta \in (0,\pi) \setminus N$, $\mathrm{Leb}(N) =
0$, there is $C_\beta > 0$ such that
$$
C_\beta K_1(\beta, \cdot) = K_2(\beta,\cdot) \tilde f(\cdot)^2.
$$

Fix $\beta \in [0,\pi) \setminus N$ and let
$$
M_\beta := \{ \alpha : \lambda(\beta,\alpha) \in \mathrm{supp}\, r
\}.
$$
Then, for almost every $\alpha \in M_\beta$, we have $C_\beta =
R_+^2(-1,\beta,\alpha)\tilde f(\alpha)^2$, or
\begin{equation}\label{14}
\tilde f(\alpha)^2 = C_\beta
R_{-1}^2(0,\beta,\lambda(\beta,\alpha)).
\end{equation}
Let $[A,B]$ be a non-trivial interval that is contained in the
support of $r$ (recall that $r$ is continuous). If $c_\beta$ and $d_\beta$ are the unique phases
determined by $\lambda(\beta,c_\beta) = B$ and
$\lambda(\beta,d_\beta) = A$ (i.e., $c_\beta =
\varphi_{-1}(0,\beta,B)$ and $d_\beta = \varphi_{-1}(0,\beta,A)$),
then $c_\beta < d_\beta$ and $[c_\beta,d_\beta] \subset M_\beta$.
Moreover, $c_\beta$ and $d_\beta$ are strictly increasing and
continuous in $\beta$ (cf.~Lemma~\ref{lem:thetader}) and we have $[c_{\beta + \pi},d_{\beta +
\pi}] = [c_\beta + \pi , d_\beta + \pi]$.

It follows that
$$
I := \bigcup_{\beta \in [0,\pi) \setminus N} ( c_\beta , d_\beta )
$$
is an open interval of length greater than $\pi$. For fixed $\beta
\in (0,\pi) \setminus N$, $\tilde f^2$ is real-analytic on
$(c_\beta , d_\beta)$ by \eqref{14}. This uses that (i) $\alpha(\lambda)$ is analytic in $\lambda$ with $\alpha'(\lambda)<0$ (and thus its inverse function $\lambda(\beta,\alpha)$ is analytic in $\alpha$) and (ii) $R_{-1}^2(0,\beta,\lambda)$ is analytic in $\lambda$. Property (ii) follows from part (a) of Lemma~\ref{lem:lamderphi}. This also implies the analyticity of the right hand side of \eqref{e.alphalambda} in $\lambda$ which in turn gives (i).

We conclude
that $\tilde f^2$ is analytic on $I$ and, due to
$\pi$-periodicity, on all of $\R$.

Now, we again fix a $\beta \in [0,\pi) \setminus N$ and conclude
by analytic continuation that \eqref{14} holds for all $\alpha$
for which $\lambda(\beta,\alpha)$ exists. Therefore,
$R_{-1}^2(0,\beta,\lambda(\beta,\alpha))$ is bounded in $\alpha$
(as this holds for the $\pi$-periodic analytic function on the LHS
of \eqref{14}). But $\lambda(\beta,\alpha)$ takes on arbitrary
real values as $\alpha$ varies and hence
\begin{equation}\label{prampbdd}
\sup_{\lambda \in \R} R_{-1}(0,\beta,\lambda) < \infty,
\end{equation}
which is impossible by Proposition~\ref{proposition2}. This
contradiction completes the proof of $\|T_1\| = \| L_0 \| < 1$.
\end{proof}

\section{The Dependence of $T_1$ on the Background} \label{s.T122continuity}

In this section we study the map $(g,E) \mapsto T_1(g,E)$. Note that the energy $E$ can be absorbed in $g$, that is, $T_1(g,E) = T_1(g-E,0)$. For this reason, we will consider without loss of generality the case $E = 0$. Consequently, in this section, $E$ is dropped from the notation and assumed to be zero. For example, we write $T_1(g)$ for $T_1(g,0)$ and $\lambda(\beta,\alpha,g)$ for $\lambda(\beta,\alpha,g,0)$.

Write
$$
D(g) = \{ (\beta,\alpha) \in \R^2 : \lambda(\beta,\alpha,g) \text{
exists} \}
$$
and
$$
A(g) = \lambda(\cdot,\cdot,g)^{-1}([-M,M]) \subset D(g),
$$
where, as in Section~\ref{s.reduction}, supp$\,r \subset [-M,M]$.

\begin{lemma}
Suppose $g_n \to g$ in $L^\infty(-1,0)$. Then, we have
\begin{equation}\label{dincl}
D(g) \subseteq \liminf_{n \to \infty} D(g_n)
\end{equation}
and
\begin{equation}\label{aincl}
\R^2 \setminus A(g) \subseteq \liminf_{n \to \infty} \R^2
\setminus A(g_n).
\end{equation}
\end{lemma}

\begin{proof}
Let $(\beta,\alpha) \in D(g)$ so that $\bar \lambda :=
\lambda(\beta,\alpha,g)$ exists. Fix some $\varepsilon > 0$. Then,
by monotonicity,
$$
\varphi_0(-1,\alpha,\bar \lambda - \varepsilon, g) < \beta <
\varphi_0(-1,\alpha,\bar \lambda + \varepsilon, g).
$$
It follows from $L^1_\mathrm{loc}$-continuity of solutions in $g$, specifically the bound provided in Lemma~\ref{lem:contdep}, that for
$n$ sufficiently large,
$$
\varphi_0(-1,\alpha,\bar \lambda - \varepsilon, g_n) < \beta <
\varphi_0(-1,\alpha,\bar \lambda + \varepsilon, g_n).
$$
Thus, for such values of $n$, there is $\bar \lambda_n \in (\bar
\lambda - \varepsilon , \bar \lambda + \varepsilon)$ with
$\varphi_0(-1,\alpha,\bar \lambda_n,g_n) = \beta$. In particular,
$\lambda(\beta,\alpha,g_n)$ exists (and is given by $\bar
\lambda_n$). This proves \eqref{dincl}. For later use, we note
that the proof also shows $\lambda(\beta,\alpha,g_n) \to
\lambda(\beta,\alpha,g)$.

Now consider $(\beta,\alpha) \in \R^2 \setminus A(g)$. That is,
either $\lambda(\beta,\alpha,g)$ does not exist or it does exist
but lies outside the interval $[-M,M]$. Suppose there is a sequence
$n_k \to \infty$ such that $\lambda(\beta,\alpha,g_{n_k})$ exists
and belongs to $[a,b]$ for every $k$. This means that
$\varphi_{-1}(0,\beta,\lambda(\beta,\alpha),g_{n_k}) = \alpha$. By
monotonicity, this gives
$$
\varphi_{-1}(0,\beta,b,g_{n_k}) \le \alpha \quad \text{and} \quad
\varphi_{-1}(0,\beta,a,g_{n_k}) \ge \alpha.
$$
Taking $k \to \infty$, we find
$$
\varphi_{-1}(0,\beta,b,g) \le \alpha \quad \text{and} \quad
\varphi_{-1}(0,\beta,a,g) \ge \alpha.
$$
This means, however, that there exists $\lambda \in [-M,M]$ such
that $\varphi_{-1}(0,\beta,\lambda,g) = \alpha$, which is a
contradiction. This proves \eqref{aincl}.
\end{proof}

\begin{lemma}\label{contkerlem}
Suppose $g_n \to g$ in $L^\infty(-1,0)$. Then,
$$
\lim_{n \to \infty} T_1(\beta,\alpha,g_n) = T_1(\beta,\alpha,g).
$$
for every $(\beta,\alpha) \in \R^2$.
\end{lemma}

\begin{proof}
We first consider the case $(\beta,\alpha) \in \R^2 \setminus
A(g)$. As seen above, this implies $(\beta,\alpha) \in \R^2
\setminus A(g_n)$ for $n \ge N_1$. Consequently,
$T_1(\beta,\alpha,g) = T_1(\beta,\alpha,g_n) = 0$ for $n \ge N_1$,
which trivially implies convergence.

If $(\beta,\alpha) \in D(g)$, we know that $(\beta,\alpha) \in
D(g_n)$ for $n \ge N_2$. Then, using continuous dependence of solutions on the potential again, it is readily
seen that
$$
T_1(\beta,\alpha,g_n) =
\frac{R_0(-1,\alpha,\lambda(\beta,\alpha,g_n),g_n)
r(\lambda(\beta,\alpha,g_n))}{\int f(x)
u_0^2(x,\alpha,\lambda(\beta,\alpha,g_n),g_n) \, dx} \to
T_1(\beta,\alpha,g).
$$
Here we also used that $\lambda(\beta,\alpha,g_n) \to
\lambda(\beta,\alpha,g)$, which was proven above.
\end{proof}

\begin{prop}
The real-valued map $g \mapsto \|T_1(g)\|_{2,2}$ is continuous on the ball of radius $\|W_0\|_{\infty}$ in $L^{\infty}(-1,0)$.
\end{prop}

\begin{proof}
We have to show that for $g,g_n \in L^\infty(-1,0)$, $\|g\|_{\infty}, \|g_n\|_{\infty} \le \|W_0\|_{\infty}$, $n\ge 1$,
with $\|g_n - g\|_\infty \to 0$, we have
$$
\lim_{n \to \infty} \|T_1(g_n)\|_{2,2} =
\|T_1(g)\|_{2,2}.
$$

By Lemma~\ref{intdeclem}.(c), it suffices to show that
\begin{equation}\label{t2toshow}
\lim_{n \to \infty} \| L_0(g_n) \|_{2,2} = \| L_0(g) \|_{2,2}.
\end{equation}

Recall that
$$
L_0 (\beta,\alpha,\cdot) = \sum_{n = 0}^{2N-1} T_1(\beta,\alpha + \pi n,\cdot);
$$
compare Lemma~\ref{intdeclem}.(b). Using (\ref{t1kerdef}) and the a priori bounds Lemma~\ref{lem:solest} and Lemma~\ref{lem:L2b}, this
implies that $L_0 (\beta,\alpha,\cdot)$ is uniformly bounded,
uniformly for $\{g\} \cup \{ g_n \}_{n \ge 1}$. By Lemma~\ref{contkerlem}, the
functions $L_0 (\cdot,\cdot,g_n)$ converge pointwise to $L_0
(\cdot,\cdot,g)$. Consequently,
\begin{align*}
\| L_0(g_n) - L_0(g) \|^2_{2,2} & \le \|
L_0(g_n) - L_0(g) \|^2_\mathrm{HS} \\
& = \int_0^\pi \int_0^\pi \left| L_0(\beta,\alpha,g_n) -
L_0(\beta,\alpha,g) \right|^2 \, d\alpha \, d\beta \\
& \to 0
\end{align*}
by dominated convergence. This proves \eqref{t2toshow} and hence
the theorem.
\end{proof}

\begin{appendix}

\section{Large Coupling Limit of the Pr\"ufer Amplitude} \label{s.appA}

Here we establish a technical fact which was used in the proof of Proposition~\ref{t1normthm}.

\begin{prop} \label{proposition2}
It holds that
\begin{equation}\label{rlaminf}
\lim_{\lambda \to \infty} R_{-1}(0,\beta,\lambda) = \infty.
\end{equation}
\end{prop}

\begin{proof}
For $[a,b]$ from \eqref{singlesite2} let $\theta \in [0,\pi)$ be such
that $\theta = \varphi_{-1}(a,\beta,\lambda,g) \mod \pi$ and
denote by $\varphi(x,\lambda) := \varphi_a(x,\theta,\lambda,g)$
and $R(x,\lambda) := R_a(x,\theta,\lambda,g)$ the Pr\"ufer phase
and amplitude for the solution of $-u''+gu+\lambda fu=0$ with
$u'(a)=\cos\theta$, $u(a)=\sin\theta$. It suffices to show that
\begin{equation} \label{prop2I}
\lim_{\lambda\to \infty} R(b,\lambda) = \infty.
\end{equation}
This follows as supp$\,f \cap ([-1,a)\cup (b,0]) =\emptyset$ and
therefore, by Lemma~\ref{lem:solest}, $R_{-1}(a,\alpha,\lambda,g) \approx 1$ and
$R_{-1}(0,\alpha,\lambda,g) \approx R_{-1}(b,\alpha,\lambda,g)
\approx R(b,\lambda)$.

Note that $\varphi$ and $R$ satisfy the Pr\"ufer differential
equations
\begin{equation} \label{prop2II}
\varphi' = 1-(1+g+\lambda f)\sin^2 \varphi
\end{equation}
and
\begin{equation} \label{prop2III}
(\ln R)' = \frac{1}{2} (1+g+\lambda f) \sin 2\varphi.
\end{equation}

We will first show that there exists $\lambda_0 \in \R$ and $\eta
\in (0,\pi)$ such that
\begin{equation} \label{prop2IV}
\varphi(x,\lambda) < \eta \quad \mbox{for all } \lambda \ge
\lambda_0 \mbox{ and all } x\in [a,b].
\end{equation}

By \eqref{prop2II} and $\theta\ge 0$ we know that
$\varphi(b,\lambda)>0$ for all $\lambda$. Sturm comparison Lemma~\ref{lem:sturm} or, more directly, Lemma~\ref{lem:thetader} shows that in proving \eqref{prop2IV} it suffices to assume that $\theta
\in [\pi/2, \pi)$. Choose $\eta \in (\theta, \pi)$ with $\sin^2
\eta = \frac{1}{2} \sin^2\theta$ and let
$$
M_{\lambda} := \Big\{ x\in[a,b]: 1+g(x)+\lambda f(x) \ge
\frac{1}{\sin^2 \eta} \Big\}.
$$
It follows from \eqref{singlesite2} that $|M_{\lambda}| \to b-a$ as
$\lambda\to\infty$. To show \eqref{prop2IV} for the given choice
of $\eta$, we assume, by way of contradiction, that there are
arbitrarily large $\lambda>0$ for which the set $\{x\in [a,b]:
\varphi(x,\lambda)\ge \eta\}$ is non-empty and thus has a minimum
$b_{\lambda}$ with $\varphi(b_{\lambda},\lambda)= \eta$. Also, let
$a_{\lambda} := \max \{x\in [a,b_{\lambda}]: \varphi(x,\lambda) =
\theta\}$. Thus $\varphi(x,\lambda) \in [\theta,\eta]$ for all
$x\in [a_{\lambda}, b_{\lambda}]$.

By \eqref{prop2II} we have $\varphi'(x) \le 1+\|g\|_{\infty}$ for
all $x\in [a_{\lambda}, b_{\lambda}]$ and $\varphi'(x)\le 1 -
\sin^2 \phi(x)/\sin^2 \eta \le 0$ for $x\in M_{\lambda} \cap
[a_{\lambda}, b_{\lambda}]$.  Thus
\begin{eqnarray*}
\eta - \theta & = & \varphi(b_{\lambda},\lambda) - \varphi(a_{\lambda},\lambda) = \int_{a_{\lambda}}^{b_{\lambda}} \varphi'(x,\lambda)\,dx \\
& \le & \int_{[a_{\lambda},b_{\lambda}] \setminus M_{\lambda}}
(1+\|g\|_{\infty})\,dx \le (1+\|g\|_{\infty}) (b-a-|M_{\lambda}|).
\end{eqnarray*}
Choosing a sufficiently large $\lambda > 0$, we can make the
right-hand side arbitrarily small and hence we obtain the
contradiction $\eta-\theta \le 0$, proving \eqref{prop2IV}.

Next, consider the set
$$
N_{\lambda} := \Big\{ x\in[a,b]:\, \varphi(x,\lambda)\ge
\frac{\pi}{4}\Big\}.
$$
As $f\ge 0$, we see by Lemma~\ref{lem:lamderphi}(b) that $N_{\lambda}$ is
decreasing for increasing $\lambda$. We will show that
\begin{equation} \label{prop2VI}
\lim_{\lambda\to\infty} |N_{\lambda}| = 0
\end{equation}
and
\begin{equation} \label{prop2VII}
\sup_{\lambda>0} \lambda \int_{N_{\lambda}} f(x)\,dx < \infty.
\end{equation}

By \eqref{prop2II} we have for all $\lambda$ that
\begin{equation} \label{prop2VIII}
\pi > |\varphi(b,\lambda)-\varphi(a,\lambda)| = \left|b-a-\int_a^b
(1+g(x)+\lambda f(x)) \sin^2 \varphi(x,\lambda)\,dx \right|.
\end{equation}
For $\lambda\ge \lambda_0$ from \eqref{prop2IV} we can bound
$$
\int_a^b (1+g(x)+\lambda f(x))\sin^2\varphi(x,\lambda)\,dx \ge
|b-a|(1-\|g\|_{\infty}) +\lambda \min\{\sin^2\eta,1/2\}
\int_{N_{\lambda}} f(x)\,dx.
$$
By \eqref{prop2VIII} it follows that $\lambda \int_{N_{\lambda}}
f(x)\,dx$ is bounded in $\lambda$, proving \eqref{prop2VII}. This
implies \eqref{prop2VI} as $|N_\lambda| \to |N|$, where $N=
\bigcap_{\lambda} N_{\lambda}$ and $f>0$ almost everywhere on $N$.

We will now use \eqref{prop2III} to prove \eqref{prop2I}. We
estimate
\begin{align} \label{prop2IX}
\int_a^b (1+g(x)+\lambda f(x))\sin 2\varphi(x,\lambda)\,dx \ge &
-(1+\|g\|_{\infty})(b-a) \\
& +\lambda \int_a^b f(x) \sin 2\varphi(x,\lambda)\,dx. \nonumber
\end{align}
By \eqref{prop2VII},
\begin{equation} \label{prop2X}
\lambda \int_a^b f(x) \sin 2\varphi(x,\lambda)\,dx \ge -\lambda
\int_{N_{\lambda}} f(x)\,dx \ge -C
\end{equation}
uniformly in $\lambda$. Moreover, Sturm comparison (Lemma~\ref{lem:sturm}) with the
solution $u(x)=\exp(\sqrt{K}x)$ of $-u''+Ku=0$, where
$K=\lambda\|f\|_{\infty} +\|g\|_{\infty}$, gives for
$\theta\not=0$ and $\lambda$ sufficiently large (such that
$\theta>1/\sqrt{K}$) that $\varphi(x,\lambda)\ge 1/\sqrt{K}$. This
shows, using the definition of $N_{\lambda}$ and \eqref{prop2VI},
\begin{equation} \label{prop2XI}
\lambda \int_{[a,b]\setminus N_{\lambda}} f(x) \sin
2\varphi(x,\lambda) \gtrsim \sqrt{\lambda} \int_{[a,b]\setminus
N_{\lambda}} f(x)\,dx \gtrsim \sqrt{\lambda}.
\end{equation}
Finally, \eqref{prop2IX}, \eqref{prop2X} and \eqref{prop2XI} yield
$$
\ln R(b,\lambda) = \int_a^b (1+g(x)+\lambda f(x))\sin
2\varphi(x,\lambda)\,dx \to \infty \quad \mbox{as
$\lambda\to\infty$}.
$$
Thus we have shown \eqref{prop2I} for $\theta\not=0$. The case
$\theta=0$ is slightly different. Comparing with the solution of
$-u''+Ku=0$, $u(a)=0$, in this case gives, for suitable $C_1>0$
and $C_2>0$,
$$
\varphi(x,\lambda) \ge \frac{C_1}{\lambda} \quad \mbox{if } |x-a|
\ge \frac{C_2}{\sqrt{\lambda}}.
$$
This suffices to get the bound \eqref{prop2XI} and thus lets the
above argument go through.

\end{proof}

\section{Pr\"ufer Variables and A Priori Bounds} \label{s.appB}

This appendix contains a brief summary of standard facts on Pr\"ufer variables and a priori bounds on solutions which have been used throughout the paper. Proofs can for example be found in the appendix of \cite{hss} (for Lemma~\ref{lem:solest}, Lemma~\ref{lem:L2b}, Lemma~\ref{lem:xder}, Lemma~\ref{lem:lamderphi}(b) and Corollary~\ref{cor:Ederphi}) and the appendix of \cite{poisson} (Lemma~\ref{lem:thetader}). Lemma~\ref{lem:lamderphi}(a) and Lemma~\ref{lem:sturm} are special cases of Theorems~2.1 and 13.1 in \cite{weidmann2}. Lemma~\ref{lem:contdep} is proven as Lemma~A.2 in \cite{dss} for the special case that $u_1$ and $u_2$ satisfy the same initial condition at $y$. The proof given there extends easily to give the result we need here.

\begin{lemma} \label{lem:solest}
For every $q\in L^1_{\rm loc}(\R)$, every solution $u$ of $-u''+qu=0$, and all $x,y \in \R$ one has
\begin{eqnarray*}
\lefteqn{(|u(y)|^2 +|u'(y)|^2) \exp \left(- \int_{\min(x,y)}^{\max(x,y)} |1+q(t)|\, dt  \right)} \nonumber  \\ & \leq &  |u(x)| ^2 + |u'(x)|^2 \leq \left( |u(y)|^2 + |u'(y)|^2 \right)
\exp \left( \int_{\min(x,y)}^{\max(x,y)} |1+q(t)|\, dt  \right).
\end{eqnarray*}
\end{lemma}

\begin{lemma} \label{lem:contdep}
For $i=1,2$, let $q_i \in L^1_{\rm loc}(\R)$ and let $u_i$ be solutions of $-u_i''+q_iu_i=0$. Then for any $x\in \R$,
\begin{eqnarray*}
\lefteqn{\left( |u_1(x)-u_2(x)|^2 + |u_1'(x)-u_2'(x)|^2 \right)^{1/2}} \nonumber \\
& \le & \left( |u_1(y)-u_2(y)|^2 + |u_1'(y)-u_2'(y)|^2 \right)^{1/2} \exp\Big\{ \int_{\min(x,y)}^{\max(x,y)} (|q_2(t)|+1)\,dt\Big\} \nonumber \\
& & \mbox{} + \left( |u_1(y)|^2 + |u_1'(y)|^2 \right) \exp\Big\{ \int_{\min(x,y)}^{\max(x,y)} (|q_1(t)|+|q_2(t)|+2)\,dt\Big\} \nonumber \\ & & \mbox{} \times \int_{\min(x,y)}^{\max(x,y)} |q_1(t)-q_2(t)|\,dt.
\end{eqnarray*}
\end{lemma}

\begin{lemma} \label{lem:L2b}
For any positive real numbers $\ell$ and $M$ there exists $C>0$ such that
\begin{equation} \label{eq:intbd2}
\int_c^{c+\ell} |u(t)|^2 dt \geq C \left( |u(c)|^2 + | u'(c)|^2 \right)
\end{equation}
for every $c\in \R$, every $L^1_{\rm loc}$-function $q$ with $\int_c^{c+\ell} |q(t)|\,dt \le M$, and any solution $u$ of $-u''+qu=0$ on $[c,c+\ell]$.
\end{lemma}

Our remaining results relate to Pr\"ufer variables.
In general, for any real potential $q\in L^1_{\rm loc}( \mathbb{R})$ and
real parameters $c$ and $\theta$ let $u_c$ be
the solution of
\[
-u'' + q u = 0
\]
with $u_c(c)=\sin \theta$, $u_c'(c)=\cos \theta$. By regarding this
solution and its derivative in polar coordinates, we define the
Pr\"ufer amplitude $R_c(x)$ and the Pr\"ufer phase
$\phi_c(x)$ by writing
\begin{equation} \label{eq:pruferdef2}
u_c(x)= R_c(x) \sin \phi_c(x)
\quad \mbox{and} \quad u_c'(x) = R_c(x) \cos
\phi_c(x).
\end{equation}
For uniqueness of the Pr\"ufer phase we declare
$\phi_c(c)=\theta$ and require continuity of $\phi_c$ in
$x$. In what follows the initial phase $\theta$ will be fixed and we thus leave the dependence of $u_c$, $R_c$ and $\phi_c$ on $\theta$ implicit in our notation.

In the new variables $R$ and $\phi$ the second order equation $-u''+qu=0$ becomes a system of two first order equations, where the equation for $\phi$ is not coupled with $R$:

\begin{lemma} \label{lem:xder} For fixed $c$ and $\theta$, one has
  that
\begin{equation} \label{eq:xderR}
(\ln R^2_c(x))'  \, = \, \left( 1 + q(x) \right) \, \sin \left( 2 \,
\phi_c(x) \right),
\end{equation}
and
\begin{equation} \label{eq:xderphi}
\phi_c'(x) \, = \, 1 \, -
\, \left( 1 + q(x)\right) \, \sin^2 \left( \phi_c(x) \right).
\end{equation}
\end{lemma}

When considering $\phi_c(x)$ at fixed $x$ as a function of the initial phase $\theta$ one can show

\begin{lemma} \label{lem:thetader}
For fixed $c$ and $x$, one has
\begin{equation} \label{eq:thetader}
\frac{\partial}{\partial \theta} \phi_c (x,\theta) = \frac{1}{R_c^2(x,\theta)}.
\end{equation}
\end{lemma}

Next we provide some results about the dependence of solutions on a coupling constant at a potential.

\begin{lemma} \label{lem:lamderphi} Let $W$ and $V$ be real valued
  functions in $L^1_{\rm loc}( \mathbb{R})$. For real parameters $c$, $\theta$
  and $\lambda$, let $u_c(\cdot,\lambda)$ be the solution of
\[
-u'' +  Wu + \lambda V u = 0,
\]
normalized so that $u_c(c,\lambda) = \sin \theta$ and $u_c'(c,\lambda) = \cos
\theta$. Denote the Pr\"ufer variables of $u_c(\cdot,\lambda)$ and $u_c'(\cdot,\lambda)$ by $\phi_c(x,\lambda)$ and $R_c(x,\lambda)$.

{\rm (a)} For fixed $x$, $u_c(x,\lambda)$ and $u_c'(x,\lambda)$ (and thus also $\phi_c(x,\lambda)$ and $R_c(x,\lambda)$) are entire functions of $\lambda$.

{\rm (b)} One has that
\begin{equation} \label{eq:lamder}
\frac{\partial}{\partial \lambda} \phi_c(x,\lambda) \, = \,
-R_c^{-2}(x,\lambda) \, \int_c^x V(t) \, u_c^2(t,\lambda) \, dt.
\end{equation}
\end{lemma}

As a special case one finds the energy derivative of the Pr\"ufer
phase.

\begin{coro} \label{cor:Ederphi}
Let $u$ be the solution of $-u'' +  Wu = Eu$
normalized so that $u(c) = \sin \theta$ and $u'(c) = \cos
\theta$, and let $\phi_c(x,E)$ and $R_c(x,E)$ be the corresponding Pr\"ufer variables. Then
\begin{equation} \label{eq:Ederphi}
\frac{ \partial}{ \partial E} \phi_c(x,E) \, = \,
R_c^{-2}(x,E) \, \int_c^x u^2(t) \, dt.
\end{equation}
\end{coro}

Finally, we state a version of Sturm's comparison theorem.

\begin{lemma} \label{lem:sturm}
For $i=1,2$, let $u_i$ be the solution of $-u_i''+q_iu_i=0$ with $u_i(c)=\sin \theta_i$ and $u_i'(c)=\cos \theta_i$. Define the Pr\"ufer phases $\phi_i(x)$ to $(u_i,u_i')$ as in (\ref{eq:pruferdef2}).

Suppose that $q_1(t)\ge q_2(t)$ for all $t\in [c,x]$ and $\theta_2 \ge \theta_1$, then $\phi_2(x) \ge \phi_1(x)$.
\end{lemma}

\end{appendix}


\begin{thebibliography}{99}

\bibitem{aenss} M.\ Aizenman, A.\ Elgart, S.\ Naboko, J.\ Schenker, G.\ Stolz, Moment analysis of localization in random
Schr\"odinger operators, \textit{Invent.\ Math.}\ \textbf{163}
(2006), 343--413

\bibitem{cfks} H.\ Cycon, R.\ Froese, W.\ Kirsch, B.\ Simon, \textit{Schr\"odinger Operators
with Applications to Quantum Mechanics and Global Geometry}, Texts
and Monographs in Physics, Springer-Verlag, Berlin, 1987

\bibitem{dss} D.\ Damanik, R.\ Sims, G.\ Stolz, Localization for one-dimensional, continuum,
Bernoulli-Anderson models, \textit{Duke Math.\ J.}\ \textbf{114}
(2002), 59--100

\bibitem{desiso} F.\ Delyon, H.\ Kunz, B.\ Souillard, One-dimensional wave equations
in disordered media, \textit{J.\ Phys.\ A}  \textbf{16} (1983), 25--42

\bibitem{dss2} F.\ Delyon, B.\ Simon, B.\ Souillard, From power pure point to continuous spectrum in disordered systems, \textit{Ann.\ Inst.\ Henri Poincar\'e Phys.\ Th\'eor.}\ \textbf{42} (1985), 283--309

\bibitem{hss} E.\ Hamza, R.\ Sims, G.\ Stolz, A note on fractional moments for the one-dimensional continuum Anderson model, preprint (arXiv:0907.4771), to appear in \textit{J.\ Math.\ Anal.\ Appl.}

\bibitem{gmp} I.\ Goldsheid, S.\ Molchanov, L.\ Pastur, A pure point spectrum of the stochastic and one dimensional Schr\"odinger equation, \textit{Funct.\ Anal.\ Appl.} \textbf{11} (1977), 1--10

\bibitem{kls} A.\ Kiselev, Y.\ Last, B.\ Simon, Modified Pr\"ufer and EFGP transforms
and the spectral analysis of one-dimensional Schr\"odinger operators, \textit{Commun.\
Math.\ Phys.}\ \textbf{194} (1998), 1--45

\bibitem{ks} H.\ Kunz, B.\ Souillard, Sur le spectre des
op\'erateurs aux diff\'erences finies al\'eatoires, \textit{Commun.\
Math.\ Phys.}\ \textbf{78} (1980), 201--246

\bibitem{ku} S.\ Kotani and N.\ Ushiroya, One-dimensional Schr\"odinger Operators with Random Decaying Potentials, \textit{Commun.\ Math.\ Phys.}\ \textbf{115} (1988), 247--266

\bibitem{r} G.\ Royer, \'Etudes des op\'erateurs de Schr\"odinger
\`a  potentiel al\'eatoire en dimension $1$, \textit{Bull.\ Soc.\
Math.\ France} \textbf{110} (1982), 27--48

\bibitem{simon} B.\ Simon, Some Jacobi matrices with decaying potential and dense point
spectrum, \textit{Comm.\ Math.\ Phys.}\ \textbf{87} (1982),
253--258

\bibitem{poisson} G.\ Stolz, Localization for random Schr\"odinger operators with Poisson potential, \textit{Ann.\ Inst.\ Henri Poincar\'e} \textbf{63} (1995), 297--314

\bibitem{Weidmann} J.\ Weidmann, \textit{Linear Operators in Hilbert Spaces}, Graduate Texts in Mathematics~\textbf{68}, Springer-Verlag, New York, 1980

\bibitem{weidmann2} J.\ Weidmann, \textit{Spectral Theory of Ordinary Differential Operators}, Lecture Notes in Mathematics~\textbf{1258}, Springer, Berlin, Heidelberg 1987

\end{thebibliography}
\end{document}